\documentclass[conference,letterpaper]{IEEEtran}

\usepackage[utf8]{inputenc} 
\usepackage[T1]{fontenc}
\usepackage{url}
\usepackage{ifthen}
\usepackage{cite}

\usepackage[cmex10]{amsmath} % Use the [cmex10] option to ensure complicance
\usepackage{cleveref}
\usepackage{amsmath,amsfonts,bm}

\def\eqref#1{equation~\ref{#1}}
\def\1{\bm{1}}

\def\rvx{{\mathbf{x}}}
\def\rvy{{\mathbf{y}}}

\DeclareMathAlphabet{\mathsfit}{\encodingdefault}{\sfdefault}{m}{sl}
\SetMathAlphabet{\mathsfit}{bold}{\encodingdefault}{\sfdefault}{bx}{n}

\newcommand{\KL}{D_{\mathrm{KL}}}
\newcommand{\Var}{\mathrm{Var}}

\usepackage{amsthm}
\usepackage{xcolor}
\usepackage{mathtools}
\usepackage{pifont}
\usepackage[OT2,T1]{fontenc}

\usepackage{thmtools}
\usepackage{thm-restate}

\newtheorem{theorem}{Theorem}[section]

\newtheorem{definition}{Definition}[section]
\newtheorem{lemma}[theorem]{Lemma}
\newtheorem{corollary}{Corollary}[section]

\DeclarePairedDelimiterX{\infdivx}[2]{[}{]}{%
  #1\;\delimsize\|\;#2%
}
\DeclarePairedDelimiterX{\inftvd}[2]{(}{)}{%
  #1\;\delimsize\|\;#2%
}
\DeclarePairedDelimiterX{\infdivcolon}[2]{[}{]}{%
  #1\;\delimsize;\;#2%
}
\DeclarePairedDelimiterX{\infdivent}[1]{[}{]}{%
  #1
}

\newcommand{\Ent}{\mathbb{H}\infdivent}
\newcommand{\Ente}{\mathbb{H}_e\infdivent}
\newcommand{\DiffEnt}{h\infdivent}
\newcommand{\MI}{I\infdivcolon}
\newcommand{\MIe}{I_e\infdivcolon}
\newcommand{\KLD}{\KL\infdivx}
\newcommand{\CSD}{D_{CS}\infdivx}

\DeclarePairedDelimiter{\abs}{\lvert}{\rvert}
\DeclarePairedDelimiterX{\innerProd}[2]{\langle}{\rangle}{%
    #1,#2%
}

\def\Normal{\mathcal{N}}

\renewcommand{\Var}{\mathbb{V}}
\newcommand{\Reals}{\mathbb{R}}
\newcommand{\nonnegReals}{\Reals_{\geq 0}}

\newcommand{\Nats}{\mathbb{N}}

\newcommand{\Prob}{\mathbb{P}}

\newcommand{\Ind}{\mathbf{1}}

\newcommand{\Exp}{\mathbb{E}}

\DeclareMathOperator{\pushfwd}{\sharp}

\newcommand{\BFMSpace}{\mathcal{M}}

\newcommand{\YSpace}{\mathcal{Y}}
\newcommand{\ZSpace}{\mathcal{Z}}

\newcommand{\PMSpace}{\mathfrak{P}}
\newcommand{\Palm}{\mathcal{P}}

\newcommand{\supp}{\mathrm{supp}}

\DeclareSymbolFont{cyrletters}{OT2}{wncyr}{m}{n}
\DeclareMathSymbol{\Sha}{\mathalpha}{cyrletters}{"58}
\DeclareMathSymbol{\sha}{\mathalpha}{cyrletters}{"78}

\DeclareMathOperator{\lb}{\mathrm{lb}}

\newcommand{\ind}{\mathrel{\perp\!\!\!\perp}}
\def\Var{{\textrm{Var}}\,}

\def\gam{\ensuremath{\ln\frac{dP_{X|Y}}{dP_X}}}
\def\gamyn{\ensuremath{\frac{1}{n}\sum\limits_{i=1}^n\gam(X_i|y_i)}}

\def\empLam{\ensuremath{\Lambda}_{X}}
\def\muthird{\ensuremath{\mu_n^{(3)}}}

\newcommand\blfootnote[1]{%
  \begingroup
  \renewcommand\thefootnote{}\footnote{#1}%
  \addtocounter{footnote}{-1}%
  \endgroup
}

\begin{document}
\title{The Redundancy of Non-Singular Channel Simulation} 

% %%% Single author, or several authors with same affiliation:
% \author{%
%  \IEEEauthorblockN{Author 1 and Author 2}
% \IEEEauthorblockA{Department of Statistics and Data Science\\
%                    University 1\\
 %                   City 1\\
  %                  Email: author1@university1.edu}% }

%%% Several authors with up to three affiliations:
\author{%
  \IEEEauthorblockN{Gergely Flamich}
  \IEEEauthorblockA{Department of Engineering \\
                    University of Cambridge\\
                    Cambridge, UK\\
                    Email: gf332@cam.ac.uk}
  \and
  \IEEEauthorblockN{Sharang M.~Sriramu and Aaron B.~Wagner}
  \IEEEauthorblockA{School of Electrical and Computer Engineering\\
  Cornell University\\
  Ithaca, NY 14853 USA\\
  Email: {sms579, wagner}@cornell.edu}
}

\maketitle

%%%%%%
%% Abstract: 
%% If your paper is eligible for the student paper award, please add
%% the comment "THIS PAPER IS ELIGIBLE FOR THE STUDENT PAPER
%% AWARD." as a first line in the abstract. 
%% For the final version of the accepted paper, please do not forget
%% to remove this comment!
%%

\begin{abstract}
%sitTHIS PAPER IS ELIGIBLE FOR THE STUDENT PAPER AWARD.
Channel simulation is an alternative to quantization and entropy coding for performing lossy source coding.
Recently, channel simulation has gained significant traction in both the machine learning and information theory communities, as it integrates better with machine learning-based data compression algorithms and has better rate-distortion-perception properties than quantization.
As the practical importance of channel simulation increases, it is vital to understand its fundamental limitations.
Recently, Sriramu and Wagner \cite{sriramu2024optimal} have provided an almost-complete characterization of the redundancy of channel simulation algorithms.
In this paper, we complete this characterisation.
First, we significantly extend a result of Li and El Gamal \cite{li2018strong} and show that the redundancy of any instance of a channel simulation problem is lower bounded by the channel simulation divergence.
Second, we give two proofs that the asymptotic redundancy of simulating i.i.d. non-singular channels is lower-bounded by $1/2$: a direct approach based on the asymptotic expansion of the channel simulation divergence and one using large deviations theory.
\end{abstract}
\section{Introduction}
% Quantization is the dominant approach for performing lossy source coding in practice and has, by now, a mature theory supporting it.
% The mechanism quantization uses to discard information is by being a many-to-one operation by design; for example, rounding forces infinitely many real numbers to share the same integer representation.
% However, this presents the single greatest mathematical obstacle to applying machine learning in lossy data compression, as quantization is a non-differentiable operation, which makes applying the standard gradient-based optimisation techniques challenging.
% \par
Channel simulation provides an alternative to quantization for lossy source coding by changing the mechanism for discarding information: it randomly perturbs its input.
Formally, let Alice and Bob be two communicating parties sharing a source of infinite common randomness $Z$. 
Furthermore, let $X, Y \sim P_{X, Y}$ be dependent random variables. 
Then, channel simulation for the channel $X \to Y$ requires that given a source sample $X \sim P_X$ Alice sends some bits to Bob so that he can simulate a single $P_{Y \mid X}$-distributed sample.
\par
What is the minimum number of bits Alice needs to communicate to Bob to achieve this?
Indeed, it is not hard to show that the communication cost must be lower bounded by $\MI{\rvx}{\rvy}$ on average \cite{li2024channel}.
This question was investigated from different angles by Li and El Gamal \cite{li2018strong}, Sriramu and Wagner \cite{sriramu2024optimal}, and by Goc and Flamich \cite{goc2024channel}.
First, Li and El Gamal show for discrete channels the redundancy of channel simulation is at least as big as a new quantity they define, called the excess functional information.
However, notably their method necessarily depends on $\rvy$ being discrete, which significantly limits their analysis.
More recently, Sriramu and Wagner considered the redundancy of large-blocklength channel simulation for i.i.d.\ distributions. 
That is, they considered dependent random variables of the form $X^n, Y^n \sim P_{X^n, Y^n} = P_{X, Y}^{\times n}$ for each $n \in \Nats$ and considered the normalized redundancy of channel simulation algorithms for the channel $X^n \to Y^n$ as $n \to \infty$.
Their analysis reveals that whether a channel is singular (see \cref{sec:preliminaries} for the definition) or not has a significant bearing on the result.
Indeed, for singular channels, they show that the normalized redundancy is \emph{at most} $o\left(\frac{\lb n}{n}\right)$.
On the other hand, for non-singular channels, they show that the normalized redundancy is \emph{at most} $\frac{1}{2n}\lb n$ bits, and show that this is tight for certain discrete channels.
However, they leave open the case for non-discrete, non-singular channels.
\par
Finally, Goc and Flamich approach the problem by restricting the solution class instead of the problem class by investigating the channel simulation efficiency of \emph{causal rejection samplers} and show a tighter one-shot lower bound using the channel simulation divergence (\cref{def:csd_def}).
\par
This paper significantly extends all three of these approaches.
Concretely, our contributions are:
\begin{enumerate}
\item For an arbitrary pair of Polish random variables $X,Y$ with the joint distribution $P_{X, Y}$, we provide a tight one-shot lower bound using the channel simulation divergence for the average communication cost of simulating the channel $P_{Y|X}$.
\item For the problem, of simulating $n$ i.i.d. copies of a non-singular channel $P_{Y|X}$ for the i.i.d. source $X^n \sim P_X^{\times n}$, we show that the optimal second-order redundancy is $\frac{\lb n}{2n}$. We provide two proofs for this; the first uses the one-shot result described earlier and the other is a standalone argument based on large deviations techniques.\blfootnote{This research was supported by the US National Science Foundation under grant CCF-2306278 and by a gift from Google.}  
\end{enumerate}

% -Benefits of Gregs Proof\\
% i) One shot?\\
% ii) Weaker moment assumptions\\
% \\
% ...
% \\
% Benefits of Sharangs proof
% i) Simpler Techniques?
% ii) High level story easier?

% ...
% Common threads:\\
% i) Likelihood ratio balls playing key role\\
% ii) Concentration phenomena being the key for the logn/2 factor
The appendix with our additional calculations for the proofs is available on the arXiv \cite{flamich2025two}.
\section{Preliminaries}
\label{sec:preliminaries}
%\subsection{Notation}
%\label{sec:notation}
\textit{Notation.}
We denote the binary logarithm by $\lb$ and the natural logarithm by $\ln$.
We denote random variables by capital Roman letters, e.g.\ $X, Y$.
We denote the expectation operator by $\Exp$, the indicator function by $\Ind[\cdot]$ and for an event $\mathcal{A}$ we use the standard notation $\Prob[\mathcal{A}] = \Exp[\Ind[X \in \mathcal{A}]]$.
We denote the Shannon entropy of a discrete random variable with probability mass function $f$ by $\Ent{X} = \Exp[-\lb f(X)]$, the differential entropy of a continuous variable that admits a density $f$ by $\DiffEnt{X} = \Exp[-\lb f(X)]$, and the mutual information between two random variables by $\MI{X}{Y}$.
We use a subscript $e$ (e.g.,  $\Ente{\cdot}$, $\MIe{\cdot}{\cdot}$) to denote information-theoretic quantities measured in nats, as opposed to bits.
\par
%\subsection{Problem Setup} \label{section:Problem setup}
\textit{Problem Setup.}
Let $\mathcal{X}$ and $\mathcal{Y}$ be Polish spaces.
Let $P_X$ be a probability measure on $\mathcal{X}$ and $P_{Y|X}: \mathcal{X} \to \mathcal{Y}$ be a Markov kernel (see \cite{regcondprob}). We will assume that $\MI{X}{Y} > 0$, which guarantees that the likelihood ratio $\frac{dP_{X|Y}}{dP_X}$ is well defined (see Lemma 3.3. \cite{YWu}). We will further assume that variance of the log-likelihood ratio is bounded --- $\Var\left[ \gam \right] < \infty$.
\par
For any block length $n>0$ and an i.i.d. source $X^n \sim P_X^{\times n}$, the tuple $(f_n,g_n,Z)$ --- where $Z \ind X$ taking values on some set $\mathcal{Z}$ is the \emph{common randomness}, $f_n: \mathcal{X}^n \times \mathcal{Z} \mapsto \{0,1\}^*$ is a \emph{prefix-free} encoder and $g_n : \{0,1\}^* \times \mathcal{Z} \mapsto \mathcal{Y}^n$ is a decoder --- simulates the channel $P_{Y|X}^{\times n}$ if $(X^n, g_n(f_n(X^n,Z),Z))$ is distributed i.i.d. $P_{XY}^{\times n}$.
The \emph{rate} of the scheme is the average length of the bit string produced by $f_n$, normalized by $n$.
\par
We are interested in characterizing the minimum average rate $R_n$, defined as the infimum of the rates across all schemes that simulate $P_{Y|X}^{\times n}$. 
Sriramu and Wagner \cite{sriramu2024optimal} show that $\lim_{n \rightarrow \infty} \frac{ R_n - \MI{X}{Y}}{\lb n/n} \le \frac{1}{2}$ for \emph{non-singular} $P_{XY}$ (Definition 4, \cite{sriramu2024optimal}) ---
i.e., where it \textbf{does not} hold $P_Y$-almost surely that 
\begin{align*}
 \frac{dP_{X \mid Y}}{dP_X}(\cdot \mid Y) = \Exp\left[\frac{dP_{X \mid Y}}{dP_X}(X \mid Y) \,\,\middle\vert\,\,Y\right]   
\end{align*} 
This paper provides a matching lower bound.
%i.e., where $\frac{dP_{X|Y}}{dP_X}(x \mid y)$ is not a deterministic function of $y$ (Definition 4, \cite{sriramu2024optimal}). This paper provides a matching lower bound.
\section{One-shot lower bound}
\begin{definition}[Width function, Channel simulation divergence \cite{goc2024channel}]
\label{def:csd_def}
Let $Q \ll P$ be probability measures over the same measurable space with Radon-Nikodym derivative $r = dQ/dP$.
Then, the $P$- and $Q$-width functions of $r$ are
\begin{align*}
w_P(h) = \Prob_{X \sim P}[r(X) \geq h],  \quad 
w_Q(h) = \Prob_{X \sim Q}[r(X) \geq h],
\end{align*}
respectively.
Then, we define the channel simulation divergence of $Q$ from $P$ by
\begin{align*}
\CSD{Q}{P} = -\int_0^\infty w_P(h) \lb w_P(h) \, dh.
\end{align*}
\end{definition}
Now, we have the following result.
\begin{theorem}[A one-shot lower bound on the efficiency of channel simulation]
\label{thm:one_shot_bound}
Let $X, Y \sim P_{X, Y}$ be a pair of dependent, Polish random variables.
Then, for any common randomness $Z$ that admits a functional representation of $Y$, i.e., such that $Z \perp X$ and $Y = g(X, Z)$ for some measurable function $g$:
\begin{align}
\label{eq:one_shot_channel_simulation_lower_bound}
\Exp_{Y \sim P_Y}[\CSD{P_{X \mid Y}}{P_X}] \leq \Ent{Y \mid Z}.
\end{align}
\end{theorem}
Observe, that  \Cref{thm:one_shot_bound} is equivalent to stating a lower bound on Li and El Gamal's excess functional information.
Indeed, as a calculation in Appendix B of \cite{goc2024channel} shows, the above result is a broad generalisation of Proposition 1 of Li and El Gamal \cite{li2018strong}, which only permits discrete $Y$.
The generalised proof's main ingredient is establishing the mathematical machinery to handle the conditional Shannon entropy $\Ent{Y \mid Z}$ when $Y$ is not discrete.
The rest of the proof follows by generalising the stochastic dominance argument in \cite{li2018strong}. 
\begin{proof}
We begin with a generalised definition of conditional entropy to make sense of the right-hand side of \cref{eq:one_shot_channel_simulation_lower_bound}.
First, let $\YSpace$ be a Polish space, denoting the range of $Y$.
Then, define
\begin{align*}
\PMSpace_\YSpace = \{\pi \in \BFMSpace_\YSpace^\sharp \mid \pi(\YSpace) = 1, \pi \text{ purely atomic}\},
\end{align*}
where $\BFMSpace_\YSpace^\sharp$ denotes the space of boundedly finite measures over $\YSpace$ \cite{daley2007introduction}.
In plain words, $\PMSpace_\YSpace$ denotes the set of all discrete probability distributions we can construct by first taking a countable subset of $\YSpace$ and assigning an appropriate probability mass to each point.
Let $g_z(x) = g(x, z)$ and let $P_{Y \mid Z = z} = g_z\pushfwd P_X$ denote the regular conditional probability of $Y$ conditioned on $\{Z = z\}$.
We may assume that $P_{Y \mid Z}$ is purely atomic $P_Z$-almost surely since otherwise the right-hand side of \cref{eq:one_shot_channel_simulation_lower_bound} is infinite and the statement of the theorem is vacuously true.
Note that there is always at least one such purely atomic Markov kernel: the Poisson functional representation \cite{li2018strong} and greedy Poisson rejection sampling \cite{flamich2023gprs} are explicit, constructive examples of such $Z$ and $g$.
Let $\ZSpace$ denote the range of $Z$ and consider the map $f: \ZSpace \to \PMSpace_\YSpace$ defined by $f(z) = P_{Y \mid Z = z}$.
Observe, that $f$ is measurable by the regularity of the conditional distribution, and hence $\Palm = f \pushfwd P_Z$ defines a measure over $\PMSpace_Y$.
From this, we also have from the law of the unconscious statistician for $A \subseteq \YSpace$:
\begin{align}
\label{eq:palm_kernel_mean_measure}
P_Y(A) = \Exp_Z[\Prob[Y \in A \mid Z]] = \int_{\PMSpace_\YSpace} \pi(A) \, d\Palm(\pi).
\end{align}
Hence, we have
\begin{align*}
\Ent{Y \mid Z} &= \int_\ZSpace \int_{\YSpace} -\lb \Prob[Y = y \mid Z = z] \, dP_{Y \mid Z}(y) \, dP_Z \\
&=\int_{\PMSpace_Y} \int_Y -\lb \pi(y) \, d\pi(y) \, d\Palm(\pi).
\end{align*}
The advantage of the latter representation is that we may now apply to it Propositions 13.1.IV and 13.1.V of \cite{daley2007introduction}.
Indeed, Proposition 13.1.IV ensures that there exists a family of measures $\{\Palm_y \mid y \in \YSpace\}$ called the local Palm distributions, that allow us to ``exchange the integrals'' and obtain
\begin{align}
\label{eq:one_shot_bound_cond_ent_exchanged_integrals}
\Ent{Y \mid Z} = \int_\YSpace \int_{\PMSpace_\YSpace} -\lb \pi(y) \, d\Palm_y(\pi)\, dP_Y(y)
\end{align}
Furthermore, Proposition 13.1.V guarantees that $\Palm_y$ is supported on the set of all purely atomic probability measures on $\YSpace$ conditioned on the event that $y$ is in their support:
\begin{align*}
\forall y \in \YSpace: \supp \Palm_y = \{\pi \in \supp \Palm \mid \pi(y) > 0\}
\end{align*}
\par
Given \cref{eq:one_shot_bound_cond_ent_exchanged_integrals}, it is now sufficient to show that $P_Y$-almost surely we have
\begin{align}
\label{eq:one_shot_bound_one_shot_ineq}
\CSD{P_{X \mid Y = y}}{P_X} \leq \int_{\PMSpace_\YSpace} -\lb \pi(y) \, d\Palm_y(\pi)
\end{align}
Then, taking expectation over $Y$ then yields the desired result.
To show \cref{eq:one_shot_bound_one_shot_ineq}, let us first define some convenient notation.
For $y \in \YSpace, X \sim P_X$ and an arbitrary measurable function $\gamma$, let ${\rho_y = \frac{dP_{Y \mid X}}{dP_Y}(X \mid y)}$, let $w_y(h) = \Prob[\rho_y \geq h]$ and let ${\Exp_{\pi \mid y}[\gamma(y, \pi)] = \int_{\PMSpace_\YSpace} \gamma(y, \pi) d\Palm_y(\pi)}$.
Then, \cref{eq:one_shot_bound_one_shot_ineq} will follow from three steps:
\begin{align}
\Exp_{\pi \mid y}[-\!\lb \pi(y)] \!&=\! \Exp_{\pi \mid y}[\pi(y)]  \Exp_{\pi \mid y}[-\!\lb \pi(y)] \Big/\Exp_{\pi \mid y}[ \pi(y)] \label{eq:one_shot_bound_mult_by_1} \\
&\geq \Exp_{\pi \mid y}[-\pi(y) \lb \pi(y)] \Big/\Exp_{\pi \mid y}[ \pi(y)] \label{eq:one_shot_bound_jensen}\\
&\geq \CSD{P_{X \mid Y = y}}{P_X}. \label{eq:one_shot_bound_csd_ineq}
\end{align}
\Cref{eq:one_shot_bound_mult_by_1} follows by multiplying by $1$.
Note that this is never troublesome since we are conditioning on the event that $\pi(y) > 0$.
Second, let $\varphi(x) = -x \lb x$; then \cref{eq:one_shot_bound_jensen} follows from Jensen's inequality applied to $\varphi$.
Finally, \cref{eq:one_shot_bound_csd_ineq} follows from a generalisation of Li and El Gamal's second-order stochastic dominance argument in the proof of Proposition 1 in \cite{li2018strong}, which we present next.
\par
First, let $\mu_y = \Exp_{\pi \mid y}[\pi(y)]$.
Then, in \cref{sec:additional_computations}, following the technique of Li and El Gamal mutatis mutandis, we show the second-order dominance type result
\begin{align}
\label{eq:one_shot_bound_second_order_dominance}
\int_0^u \Prob_{\pi \mid y}[\pi(y) \leq v] \, dv \leq u - \mu_y \int_0^u w_y^{-1}(p) \, dp.
\end{align}
Now, we have
\begin{align}
\Exp&_{\pi \mid y}[\varphi(\pi(y))] \nonumber \\
&= \lb(e)\left(1-\mu_y - \int_0^1 \int_0^u \frac{\Prob_{\pi \mid y}[\pi(y) \leq v]}{u} \, dv \, du\right) \label{eq:one_shot_bound_ibp_twice_1} \\
&\geq \mu_y \lb(e)\left(-1 + \int_0^1 \int_0^u \frac{w_y^{-1}(p)}{u} \, dp \, du\right) \tag{by \cref{eq:one_shot_bound_second_order_dominance}} \\
&=\mu_y \CSD{P_{X \mid Y =y}}{P_X} \label{eq:one_shot_bound_ibp_twice_2},
\end{align}
where \cref{eq:one_shot_bound_ibp_twice_1,eq:one_shot_bound_ibp_twice_2} both follow from integrating by parts twice (see \cref{sec:additional_computations} for the precise computation); finishing the proof.
\end{proof}
Note, that $\KLD{Q}{P} \leq \CSD{Q}{P}$ by Lemma IV.1.IV of \cite{goc2024channel}.
Therefore, $\MI{X}{Y} \leq \Exp_Y[\CSD{P_{X \mid Y}}{P_X}]$, meaning that \cref{thm:one_shot_bound} provides a tighter lower bound than what was previously known.
Furthermore, Theorem V.2 of \cite{goc2024channel} shows that the lower bound is tight up to a small constant.
In the next section, we show that \cref{thm:one_shot_bound} can be further simplified in the asymptotic case to characterise the redundancy of channel simulation fully. 
\section{Proof of the Asymptotic Result Using the Channel Simulation Divergence}
\label{sec:csd_asympt_result}
\par
This section proves a converse result to Theorem 1 of \cite{sriramu2024optimal} in the non-singular case based on \cref{thm:one_shot_bound}.
In \cref{sec:large_deviation_proof}, we provide a second, direct proof using large deviation theory.
Before we state the main result of the section, we show the following result.
\begin{lemma}
\label{lemma:divergence_difference}
Let $Q \ll P$ be probability distributions with $r = dQ/dP$ and $w_P(h) = \Prob_{X \sim P}[r(X) \geq h]$.
Then, $w_P$ defines a probability density over $\nonnegReals$.
Let $H$ be a random variable whose density is $w_P$.
Then, we have
\begin{align*}
\CSD{Q}{P} - \KLD{Q}{P} = \DiffEnt{\ln H} - \lb(e).
\end{align*}
\end{lemma}
\begin{proof}
For the first part, note that $w_P$ is non-negative and integrates to one:
$\int_0^\infty w_P(h) \, dh = \Exp_{X \sim P}[r(X)] = 1$.
The second part follows by direct calculation, see \cref{sec:additional_computations}.
\end{proof}
Now, we state the main result of the section.
\begin{theorem}[Asymptotic Redundancy of Non-singular Channels]
\label{thm:asymptotic_result_thm_1}
Let $X, Y \sim P_{X, Y}$ be dependent random variables.
Assume the channel $P_{X \mid Y}$ is almost surely non-singular and let $r_Y = dP_{X \mid Y}/{dP_X}$.
For all $n \geq 1$, let $X^n, Y^n \sim P_{X, Y}^{\times n}$ and assume that $\sigma^2 = \Var\left[\lb \frac{dP_{X, Y}}{P_X \times P_Y}(X, Y)\right] < \infty$.
Then, for any sequence of common randomness $Z_n \perp X^n$ and $Y^n = g(X^n, Z_n)$, we have
\begin{align*}
\lim_{n \to \infty}\frac{\Ent{Y^n \mid Z_n} - \MI{X^n}{Y^n}}{\lb(n)} \geq \frac{1}{2}.
\end{align*}
\end{theorem}
\begin{proof}
To begin, \cref{thm:one_shot_bound} implies that
\begin{align*}
\Exp_{Y^n}[\CSD{P_{X^n \mid Y^n}}{P_{X^n}}] \leq \Ent{Y^n \mid Z_n}.
\end{align*}
Hence, our goal from now on will be to show that
\begin{align*}
\lim_{n \to \infty}\frac{\Exp_{Y^n}[\CSD{P_{X^n \mid Y^n}}{P_{X^n}}] - \MI{X^n}{Y^n}}{\lb(n)} = \frac{1}{2}
\end{align*}
For $n \geq 1$, let ${\kappa^n = \KLD{P_{X^n \mid Y^n}}{P_{X^n}} \cdot \ln 2}$ and let $H_n$ denote the random variable induced by $w_{P_{X^n}}$ from \cref{lemma:divergence_difference}.
Then, the lemma implies that the  above limit is equal to
\begin{align*}
\lim_{n \to \infty}\frac{\DiffEnt{\ln H_n \mid Y^n}}{\lb(n)} = \frac{1}{2} +\! \lim_{n \to \infty}\frac{\DiffEnt{n^{-1/2}(\ln H_n - \kappa^n) \mid Y^n}}{\lb(n)}
\end{align*}
where the above equality follows from the translation invariance and the scaling property of the differential entropy.
The rest of the proof shows that the entropy on the right-hand side converges to a constant, which means that the limit on the right-hand side vanishes, yielding the desired result.
\par
As the scaling and shift already suggest, we show that a central limit theorem holds for the sequence ${n^{-1/2}(\ln H_n - \kappa^n)}$, and it converges to a Gaussian. 
To this end, let ${s = n^{1/2}}, {b = \kappa^n}$ and $\Lambda_n(t) = w_{P_{X^n \mid Y^n}}(\exp(s \cdot t + b))$.
Then, in \cref{sec:additional_computations}, we show that
\begin{align}
\label{eq:asymptotic_result_proof_sequence_CDF_identity}
\Prob[s(\ln H_n - b) \leq t \mid Y^n] = 1 - \Lambda_n(t) + \int_{-\infty}^t\!\!\!\! e^{-s(\eta - t)} \, d\Lambda_n(\eta) 
\end{align}
Now, let us examine the behaviour of the two terms on the right-hand side.
First, note that by Fubini's theorem, we have
\begin{align*}
\lim_{n \to \infty} \int_{-\infty}^t\!\! e^{-n^{1/2}(\eta - t)} \, d\Lambda_n(\eta) =  \int_{-\infty}^t\!\! \Ind[\eta = t] \, d\Lambda_\infty(\eta) = 0
\end{align*}
where $\Lambda_\infty(\eta) = \lim_{n \to \infty}\Lambda_n(\eta)$.
Second, observe that
\begin{align*}
1 - \Lambda_n(t) = \Prob_{P_{X^n \mid Y^n}}\left[\sum_{i = 1}^n\frac{\ln r_Y(X_i \mid Y_i) - \kappa^n}{\sqrt{n}} < t \,\middle\vert\, Y^n \right]
\end{align*}
Hence, we have
\begin{align}
\label{eq:first_asymp_proof_limit_identity}
\lim_{n \to \infty} \frac{\ln H_n - \kappa^n}{\sqrt{n}} = \lim_{n \to \infty} \sum_{i = 1}^n\frac{\ln r_Y(X_i \mid Y_i) - \kappa^n}{\sqrt{n}}
\end{align}
We finish the proof by showing that the Lindeberg-Feller central limit theorem (Theorem 3.4.5, \cite{durrett2019probability}) holds for the right-hand side of the equation.
We show the following lemma in \cref{sec:additional_computations}.
\begin{lemma}
\label{lemma:lindeberg_cond}
Denote the conditional variance of $\ln r_Y$ as $\sigma^2_Y = \Var_{X \mid Y}[\ln r_Y(X)]$.
Since $P_{X \mid Y}$ is non-singular, $\sigma^2_Y > 0$ almost surely and is also finite by assumption.
Let $s_n = \sum_{k = 1}^n \sigma^2_{Y_k}$.
Furthermore, let $\kappa_Y = \Exp_{X \mid Y}[\ln r_Y(X)]$.
Then, almost surely, the following Lindeberg condition holds: $\forall \epsilon > 0$ 
\begin{align*}
\lim_{n\to \infty}\sum_{k = 1}^n \frac{\Exp[(\ln r_{Y_k} - \kappa_{Y_k})^2\Ind[\abs{\ln r_{Y_k} - \kappa_{Y_k}} \geq s_n \epsilon]]}{s_n} = 0,
\end{align*}
which, by the Lindeberg-Feller central limit theorem, implies
\begin{align*}
\sum_{i = 1}^n\frac{\ln r_Y(X_i \mid Y_i) - \kappa^n}{\sqrt{n}} \stackrel{d}{\to} \Normal(0, \Exp_Y[\sigma^2_Y]) \quad \text{as } n \to \infty
\end{align*}
\end{lemma}
\begin{proof}
The proof follows a ``standard'' argument for showing the Lindeberg condition; see \cref{sec:additional_computations} for the details.
\end{proof}
Putting the result of the lemma together with \cref{eq:first_asymp_proof_limit_identity} finishes the proof.
\end{proof}
\section{Proof of the Asymptotic Result Using Large Deviations Techniques}
\label{sec:large_deviation_proof}
Here, we provide a second proof of our asymptotic result in \cref{thm:asymptotic_result_thm_1}, which we restate in an equivalent form below.
\begin{theorem} \label{theorem:MainResultConverse}
    Consider a non-singular joint probability measure $P_{XY}$ such that $\Var\left[ \gam \right] < \infty$. Then,
    \begin{align*}
        \lim_{n \rightarrow \infty} \frac{ R_n - \MI{X}{Y}}{\lb n/n} \ge \frac{1}{2}.
    \end{align*}
\end{theorem}
\subsection{Cumulants and exponentially tilted measures}
We will use several different cumulant generating functions pertaining to the log-likelihood ratio $\gam$.
\begin{definition}[Cumulant generating function]
    For $\lambda \in \mathbb{R}$ and $y \in \mathcal{Y}$, we define the pointwise cumulant
    \begin{align*}
        \Lambda_{X}(\lambda, y) = \ln \Exp_{X}\Bigg[& \Ind\left[\frac{dP_{X|Y}}{dP_X}(X|y) > 0\right] e^{ \lambda \gam(X|y) }\Bigg].
    \end{align*}
    For $y^n \in \mathcal{Y}^n$, we extend the notation to define $\empLam(\lambda, y^n) = \frac{1}{n}\sum\limits_{i=1}^n \Lambda_X(\lambda, y_i)$.
    The first, second and third derivatives of $\empLam(\cdot, y)$ exist when $\empLam(\cdot,y)$ exists (see 2.2.24 \cite{dembo2009large}) and we denote them by $\empLam'$, $\empLam''$ and $\empLam'''$ respectively. 
\end{definition}
\begin{definition} [Tilted measures and moments]
    For any $\lambda \in \mathbb{R}$ and $y \in \mathcal{Y}$, the $\lambda$-tilted measure w.r.t. $y$, $Q^{\lambda}_{X|y}$ satisfies
    \begin{align}
        \frac{dQ^{\lambda}_{X|y}}{dP_X}(x) =& \Ind\left[\frac{dP_{X|Y}}{dP_X}(x|y) > 0\right] e^{ \lambda\gam(x|y) - \Lambda_X(\lambda,y)}. \label{eq:TiltingMeasure}
    \end{align}
    For a sequence $y^n \in \mathcal{Y}^n$, the measure $Q^{\lambda}_{X^n|y^n}$ is defined as the product of the componentwise tilting measures. We will also define the tilted variance $s_n^2(\lambda, y^n) = \sum\limits_{i=1}^n \Var_{Q^{\lambda}_{X|y_i}}\left[  \gam(X|y_i) \right]$ and the tilted third absolute moment $\muthird(\lambda, y^n) = \sum\limits_{i=1}^n \Exp_{Q^\lambda_{X|y_i}}\left[ \Big| \gam(X|y_i) \Big|^3 \right]$.
\end{definition}
\subsection{Additional Assumptions} \label{section:AdditionalAssume}
This proof requires the additional assumption that there exists $\delta>0$ such that $\Exp_{XY}\left[ \exp\left( \delta \gam \right) \right] < \infty$. \footnote{Note that this implies that $\Exp_Y\empLam(1 + \delta, Y) < \infty$ and also that all the moments of $\left|\gam(X|Y)\right|$ are finite under $P_{XY}$ because $\Exp_{XY}\left[ \exp\left( -\gam \right) \right] < \infty$.}
\subsection{Sketch of the Main Argument}
\paragraph{Initial Rate Bound}
Given a non-singular joint probability measure $P_{XY}$ and a blocklength $n$, let $(f,g,Z)$ be any scheme with rate $R$ (in nats) that simulates $P_{Y|X}^{\times n}$ exactly. Using standard information-theoretic inequalities, we can obtain
\begin{align}
    nR &\ge \Ente{f(X^n,Z)|Z} \label{eq:BeginMain}\\
    &\ge \MIe{X^n}{f(X^n,Z)|Z}\nonumber\\
    &\ge \MIe{X^n}{Y^n|Z}\nonumber\\
    &= \Exp\left[ \ln \frac{dP_{X^n|Y^n,Z}}{dP_{X^n}} \right]. \label{eq:X_mapping_reduction}
\end{align}
Now, consider the set $A(y^n,z) = \{x^n \in \mathcal{X}^n: g(f(x^n,z),z) = y^n\}$. We can assume without loss of generality that $P_{X^n}(A(y^n,z)) > 0$. It is clear that 
\begin{align*}
\frac{dP_{X^n|Y^n,Z}}{dP_{X^n}}(x^n|y^n,z) = \frac{\mathbf{1}\left[x^n \in A(y^n,z)\right]}{P_{X^n}(A(y^n,z))}.
\end{align*}
We can then substitute this in (\ref{eq:X_mapping_reduction}):
\begin{align}
    nR &\ge -\Exp_{Y^n,Z}\left[ \ln P_{X^n}(A(Y^n, Z)) \right]. \label{eq:fixedyu}
\end{align}
\paragraph{Restriction to Typical $(Y^n,Z)$ Realizations}
Our objective is to derive a tight upper bound on $P_{X^n}(A(y^n,z))$ for "typical" realizations of $y^n$ and $z$. We construct a high probability typical set of $(y^n,z)$ pairs --- $\mathcal{U}^{n}_{\text{typical}}$ with \mbox{$1 - P_{Y^n, Z}(\mathcal{U}^{n}_{\text{typical}}) = O\left( \frac{1}{n} \right)$} --- wherein some strong regularity conditions hold. Specifically, consider the \emph{conditional mean log-likelihood ratio}:
\begin{definition}
    For all $y^n \in \mathcal{Y}^n$ and $S \subset \mathcal{X}^n$ s.t. $\frac{dp_{X^n|Y^n}}{dp_{X^n}}(x^n|y^n)>0$ for all $x^n \in S$, we define the conditional mean log-likelihood ratio 
    \begin{align*}
    \imath_{y^n}\left(S\right) = \Exp\left[\frac{1}{n}\sum\limits_{i=1}^n \gam(X_i|y_i) \Big|\, X^n \in S \right].
    \end{align*}
\end{definition}
Then, one of the typicality conditions we impose restricts $\imath_{y^n}\left(A(y^n, z)\right)$ to a small interval around the channel mutual information $\MIe{X}{Y}$. 
\par
We also impose conditions on some higher order moments of $\gam$. For tilting parameters $\lambda$ within a carefully chosen operating interval $[\underline{\lambda}, \overline{\lambda}]$, these conditions allow us to obtain a uniform upper bound on $\muthird(\lambda, y^n)$ as well as uniform lower and upper bounds on $\frac{s_n^2(\lambda, y^n)}{n}$. In other words, we can find positive constants $\underline{m}_2, \overline{m}_2, \overline{m}_3$ such that $\underline{m}_2 \le \frac{s^2_n(\lambda, y^n)}{n} \le \overline{m}_2$ and $\frac{\muthird(\lambda, y^n)}{n} \le \overline{m}_3$ for all $n>0$ and $\lambda \in [\underline{\lambda}, \overline{\lambda}]$. In addition, we also define a constant $\underline{M}$ that serves to derive a uniform lower bound on a certain large-deviations probability (see Definition~\ref{def:YucelLB}, and Lemma 1~\cite{altuug2020exact})
\par
The details of the construction are described in Appendix~\ref{section:Typicality}. 
\paragraph{Reduction to Ball-Like Sets}
Next, our goal is to obtain a tight upper bound for $P_{X^n}(A(y^n,z))$. We proceed by constructing balls that are close in measure to $A(y^n,z)$ which can then be precisely characterized using refined large-deviation bounds.
\par
We shall replace $A(y^n,z)$ with the ball
\begin{align*}
B(y^n,z) = \left\{x^n: \frac{1}{n}\sum\limits_{i=1}^n\gam(x_i|y_i) \ge \imath_{y^n,z}\right\},
\end{align*}
where we have defined the radius
\begin{equation} 
\imath_{y^n,z} = \imath_{y^n}(A(y^n,z)) - \frac{1}{\underline{M} \cdot \underline{\lambda}^2 \cdot \sqrt{\underline{m}_2} \cdot n}. \label{eq:mathidef}
\end{equation}

In words, $B(y^n,z)$ is an information-ball centered at $y^n$ such that the average
of the log-likelihood ratio is close to, but smaller than, that of the original set $A(y^n,z)$. More precisely, we have the following relation.

\begin{lemma}\label{lemma:Gibbs}
For sufficiently large $n$ and typical $(y^n,z)$ we have
\begin{align*}
\imath_{y^n}(B(y^n,z)) \le \imath_{y^n}(A(y^n,z)).
\end{align*}
\end{lemma}
\begin{proof}
    Please see Appendix~\ref{section:GibbsProof}.
\end{proof}

From this fact it follows that we can replace $A(y^n,z)$ with 
$B(y^n,z)$, as shown in the next lemma.

\begin{lemma}
For any typical $(y^n,z)$ we have
\begin{align*}
P_{X^n}(B(y^n,z)) \ge P_{X^n}(A(y^n,z)).
\end{align*}
\end{lemma}

\begin{proof}
Writing $A$ for $A(y^n,z)$, and similarly for $B$, by the previous lemma we have
$\imath_{y^n}(B) \le \imath_{y^n}(A)$, which implies
\begin{align}
    &\imath_{y^n}(B \cap A)\frac{P_{X^n}(B \cap A)}{P_{X^n}(B)} + \nonumber\\
    &\imath_{y^n}(B \cap A^c)\left(1 - \frac{P_{X^n}(B \cap A)}{P_{X^n}(B)}\right)  \label{eq:ball-construction} \\
    &\le \imath_{y^n}(B \cap A)\frac{P_{X^n}(B \cap A)}{P_{X^n}(A)} + \nonumber\\
    &\phantom{11,}\imath_{y^n}(B^c \cap A)\left(1 - \frac{P_{X^n}(B \cap A)}{P_{X^n}(A)}\right). \label{eq:convex_comb_ineq}
\end{align}
From the definition of $B$, we know that $\imath_{y^n}(B^c \cap A) \le \imath_{y^n}(B \cap A^c)$ and that $\imath_{y^n}(B^c \cap A) \le \imath_{y^n}(B \cap A)$. Therefore, (\ref{eq:convex_comb_ineq}) implies that 
\begin{align*}
1 - \frac{P_{X^n}(B \cap A)}{P_{X^n}(A)} \le 1 - \frac{P_{X^n}(B \cap A)}{P_{X^n}(B)},
\end{align*}
from which  the conclusion follows.
\end{proof}
\paragraph{Change of Measure and Conclusion}
We have, for sufficiently large $n$ and typical $(y^n,z)$ pairs,
\begin{align}
   & P_{X^n}(A(y^n, z)) \nonumber\\
     & \le  P_{X^n}(B(y^n, z)) \nonumber\\
     & =  \int \Ind(x^n \in B(y^n,z)) dP_{X^n}(x^n) \nonumber\\
     & =  \int e^{\ln \frac{dP_{X^n}}{dP_{X^n|Y^n}}(x^n,y^n)} \Ind(x^n \in B(y^n,z)) dP_{X^n|Y^n}(x^n|y^n) \nonumber\\
     & =  e^{-n \cdot \imath_{y^n,z}} \int \Big(e^{\ln \frac{dP_{X^n}}{dP_{X^n|Y^n}}(x^n,y^n) + n \cdot \imath_{y^n,z}} \nonumber\\
     &\quad\quad\quad\quad\quad\quad\,\,\,\,
     \cdot\mathbf{1}(x^n \in B(y^n,z)) dP_{X^n|Y^n}(x^n|y^n)\Big). \label{eq:PauseMain}
\end{align}
By the definition of $B(y^n,z)$ and the fact that $(y^n,z)$ is typical, the integral is upper bounded by (Lemma 2, ~\cite{altuug2020exact})
\begin{align*}
\frac{1}{\sqrt{n}} \left( \frac{1}{\sqrt{2 \pi \underline{m}_2}} + \frac{\overline{m}_3}{\underline{m}_2} \right).
\end{align*}
(see also Lemma 47~\cite{polyanskiy2010channel}). Thus continuing from (\ref{eq:BeginMain})-(\ref{eq:PauseMain}),
\begin{align}
nR  
&\ge -\Exp_{Y^n,Z}\left[ \ln P_{X^n}\left( A(Y^n,Z) \right)|\mathcal{U}^{n}_{\text{typical}} \right] \cdot P_{Y^n,Z}(\mathcal{U}^{n}_{\text{typical}}) \nonumber\\
& \ge \left[ \Exp_{Y^n,Z}\left[  n \cdot \imath_{Y^n,Z} \right] + \frac{1}{2} \ln n + \ln\left(\frac{1}{\sqrt{2 \pi \underline{m}_2}} + \frac{\overline{m}_3}{\underline{m}_2}
\right) \right] \nonumber\\
&\phantom{111} \cdot\left(1- \frac{C}{n}\right) 
\label{eq:ratemoduloradius}
\end{align}
where the characterization of the constant $C$ can be found in the proof of Lemma \ref{lemma:TypicalityProbBound} in Appendix \ref{section:Typicality}. From the definition of $\imath_{y^n,z}$ in (\ref{eq:mathidef}), we have
\begin{align}
&\Exp_{Y^n,Z}\left[  n \cdot \imath_{Y^n,Z}  \right] \nonumber\\
&= \Exp_{Y^n,Z}\left[  n \cdot \imath_{y^n}(A(Y^n,Z))\right] - \frac{1}{\underline{M} \cdot \underline{\lambda}^2 \cdot \sqrt{\underline{m}_2} \cdot n}. \label{eq:radiusexpansion}
\end{align}
Finally,
\begin{align*}
&\Exp_{Y^n,Z}\left[ n \cdot \imath(A(Y^n,Z))\right]\\
&= \Exp_{Y^n,Z} \left[  \Exp_{X^n}\left[\sum\limits_{i=1}^n \gam(X_i|Y_i) \Big|\, X^n \in A(Y^n,Z) \right]  \right] \\
& = \Exp_{X^n,Y^n} \left[  \sum\limits_{i=1}^n \gam(X_i|Y_i)\right]   \\
& = n\MIe{X}{Y}.
\end{align*}
The conclusion follows by substituting this into (\ref{eq:radiusexpansion}) and (\ref{eq:ratemoduloradius}).
\section{Discussion}
In this paper, we first proved \cref{thm:one_shot_bound}, a general, tight (by Theorem V.2 of \cite{goc2024channel}), one-shot lower bound on the rate of channel simulation algorithms in terms of the channel simulation divergence, thus establishing the importance of this information measure.
We conjecture that our bound is also achievable, but we leave proving this for future work.
\par
Second, we gave two proofs for \cref{theorem:MainResultConverse}, which, combined with Theorem 1 of \cite{sriramu2024optimal}, shows that the optimal second-order redundancy of non-singular channel simulation in the asymptotic setting is $\frac{\lb n}{2n}$.
\section{Author Contributions}
GF conceived the proofs of the one-shot result in \cref{thm:one_shot_bound} and the asymptotic result in \cref{sec:csd_asympt_result}, while SMS derived the proof of the asymptotic result using large deviations theory in \cref{sec:large_deviation_proof}.
They contributed equally to writing the paper.
ABW provided technical guidance for SMS and assisted with preparing the manuscript.

%%%%%%
%% To balance the columns at the last page of the paper use this
%% command:
%%
%\enlargethispage{-1.2cm} 
%%
%% If the balancing should occur in the middle of the references, use
%% the following trigger:
%%
%\IEEEtriggeratref{7}
%%
%% which triggers a \newpage (i.e., new column) just before the given
%% reference number. Note that you need to adapt this if you modify
%% the paper.  The "triggered" command can be changed if desired:
%%
%\IEEEtriggercmd{\enlargethispage{-20cm}}
%%
%%%%%%

%%%%%%
%% References:
%% We recommend the usage of BibTeX:
%%
%\bibliographystyle{IEEEtran}
%\bibliography{definitions,bibliofile}
%%
%% where we here have assumed the existence of the files
%% definitions.bib and bibliofile.bib.
%% BibTeX documentation can be obtained at:
%% http://www.ctan.org/tex-archive/biblio/bibtex/contrib/doc/
%%%%%%

%% Or you use manual references (pay attention to consistency and the
%% formatting style!):
\bibliographystyle{IEEEtran}
\bibliography{references}
\newpage
\onecolumn
\appendices

\crefalias{section}{appendix}

\section{Additional Computations}
\label{sec:additional_computations}
\par
\textbf{Showing \cref{eq:one_shot_bound_second_order_dominance}:}
As we noted in the main text, the proof of the inequality follows the second-order dominance technique of Li and El Gamal \cite{li2018strong}.
First, introduce for convenience the notation $(x)_+ = \max\{0, x\}$.
Then, we have
\begin{align}
\int_0^u \Prob_{\pi \mid y}[\pi(y) \leq v] \, dv
&= \Exp_{\pi \mid y}[(u - \pi(y))_+]
\label{eq:appendix_one_shot_bound_CDF_to_exp} \\
&=\Exp_{\pi \mid y}[(w_y(w_y^{-1}(u)) - \Exp_X[\pi(y \mid X)])_+] \nonumber\\
&= \Exp_{\pi \mid y}\left[\left(\Exp_{\rho_y}\left[\Ind[\rho_y \geq w_y^{-1}(u)\right] - \Exp_X[\pi(y \mid X) \mid \rho_y] \right)_+\right] \nonumber \\
&\leq \Exp_{\pi \mid y}\left[\Exp_{\rho_y}\left[\left(\Ind[\rho_y \geq w_y^{-1}(u)] - \Exp_X[\pi(y \mid X) \mid \rho_y] \right)_+\right]\right] \tag{Jensen} \\
&= \Exp_{\pi \mid y}\left[\Exp_{\rho_y}\left[\Ind[\rho_y \geq w_y^{-1}(u)] \cdot\left(1 -  \Exp_X[\pi(y \mid X) \mid \rho_y]\right) \right]\right] \label{eq:appendix_one_shot_bound_relu_identity}\\
&=\Exp_{\rho_y}\left[\Ind[\rho_y \geq w_y^{-1}(u)] \cdot\left(1 -  \Exp_X\left[\Exp_{\pi \mid y}[\pi(y \mid X) \mid \rho_y]\right]\right) \right] \tag{Fubini} \\
&= \Exp_{\rho_y}\left[\Ind[\rho_y \geq w_y^{-1}(u)] \cdot\left(1 -  \Exp_X\left[\frac{dP_{X \mid Y}}{dP_X}(X \mid y) \cdot \Exp_{\pi \mid y}[\pi(y \mid X)] \,\middle\vert\, \rho_y\right]\right) \right] \tag{Bayes' rule} \\
&= \Exp_{\rho_y}\left[\Ind[\rho_y \geq w_y^{-1}(u)] \cdot\left(1 -  \mu_y \cdot \Exp_X\left[\frac{dP_{X \mid Y}}{dP_X}(X \mid y) \,\middle\vert\, \rho_y\right]\right) \right] \nonumber \\
&= \Exp_{\rho_y}\left[\Ind[\rho_y \geq w_y^{-1}(u)] \cdot\left(1 -  \mu_y \cdot \rho_y \right) \right] \nonumber \\
&= u - \mu_y \int_0^u w_y^{-1}(v) \, dv \label{eq:appendix_one_shot_bound_last_line}
\end{align}
In the above, \cref{eq:appendix_one_shot_bound_CDF_to_exp} follows since for any random variable $V \sim P_V$ supported on $[0, 1]$ we have
\begin{align*}
\int_0^u \Prob[V \leq t] \, dt 
&= \int_0^1 \int_0^1 \Ind[v \leq t] \Ind[t \geq u] \, dt \, dP_V(v) \\
&= \int_0^1 (u - v)_+ \, dP_V(v) \\
&= \Exp_V[(u - V)_+]
\end{align*}
Next, to see that \cref{eq:appendix_one_shot_bound_relu_identity} holds, set $\pi$ and note that
\begin{align*}
\Exp_{\rho_y}[(\Ind[\rho_y \geq w_y^{-1}(u)] - \pi)_+]
&= \Exp_{\rho_y}[(\Ind[\rho_y \geq w_y^{-1}(u)] - \pi) \cdot \Ind[\Ind[\rho_y \geq w_y^{-1}(u)] \geq \pi]] \\
&= \Exp_{\rho_y}[(\Ind[\rho_y \geq w_y^{-1}(u)] - \pi) \cdot \Ind[\rho_y \geq w_y^{-1}(u)]] \tag{since $\pi \in (0, 1]$} \\
&= \Exp_{\rho_y}[(1 - \pi) \cdot \Ind[\rho_y \geq w_y^{-1}(u)]].
\end{align*}
Finally, \cref{eq:appendix_one_shot_bound_last_line} holds, because
\begin{align*}
\Exp_{\rho_y}\left[\Ind[\rho_y \geq w_y^{-1}(u)] \cdot\left(1 -  \mu_y \cdot \rho_y \right) \right] 
&= -\int_0^\infty \Ind[\rho_y \geq w_y^{-1}(u)] \cdot (1 - \mu_y \cdot \rho_y) \, dw_y(\rho_y) \\
&=\int_0^1 \Ind[w_y^{-1}(v) \geq w_y^{-1}(u)] \cdot \left(1 - \mu_y \cdot w_y^{-1}(v)\right) \, dv \\
&= \int_0^u \left(1 - \mu_y \cdot w_y^{-1}(v)\right) \, dv.
\end{align*}
\par
\textbf{Showing \cref{eq:one_shot_bound_ibp_twice_1}:} we perform a change of variables and integrate by parts twice to obtain the desired result:
\begin{align*}
\Exp_{\pi \mid y}[\varphi(\pi(y))] &= \int_{\PMSpace_\YSpace} \varphi(\pi(y)) \, d\Palm_y(\pi) \\
&= \int_0^1 \varphi(u) \, d\Prob_{\pi \mid y}[\pi(y) \leq u] \\
&= \varphi(1) - \int_0^1 \varphi'(u) \Prob_{\pi \mid y}[\pi(y) \leq u] \, du \\
&= -\varphi'(1)(1 - \mu_y) + \int_0^1 \varphi''(u) \cdot \int_0^u \Prob_{\pi \mid y}[\pi(y) \leq v] \, dv \,du \\
&= \lb(e)\left(1-\mu_y - \int_0^1 \int_0^u \frac{\Prob_{\pi \mid y}[\pi(y) \leq v]}{u} \, dv \, du\right),
\end{align*}
where we used the Darth Vader rule to obtain $\int_0^1 \Prob_{\pi \mid y}[\pi(y) \leq v] \, dv = 1 - \mu_y$ and used the facts that $\varphi(1) = 0$ and $\varphi'(1) = -\lb(e)$.

\par
\textbf{Showing \cref{eq:one_shot_bound_ibp_twice_2}:}
This result is a variation on the integral representation of the channel simulation divergence derived in Equation 2 in \cite{goc2024channel}.
We obtain the result by integrating by parts twice:
\begin{align*}
\CSD{Q}{P} 
&= \int_0^\infty \varphi(w_P(h)) \, dh \\
&= \int_{w_P(0)}^{w_P(\infty)} \varphi(p) \, dw^{-1}_P(p) \\
&= \int_0^1 \varphi'(p) w^{-1}_P(p) \, dp \\
&= -\lb(e) - \int_0^1 \varphi''(p) \int_0^p w^{-1}_P(v) \, dv \,dp \\
&= \lb(e)\left(-1 + \int_0^1 \frac{1}{p} \int_0^p w^{-1}_P(v) \, dv \, dp\right).
\end{align*}%
\par
\textbf{Showing the second part of \cref{lemma:divergence_difference}:}
By Equation 5 of \cite{goc2024channel}, we have the following integral representation of the KL divergence:
\begin{align*}
\KLD{Q}{P} = \lb(e) + \int_0^\infty w_P(h) \lb(h) \, dh.
\end{align*}
Now, using this in conjunction with the definition of the channel simulation divergence, we get
\begin{align*}
(\CSD{Q}{P} - \KLD{Q}{P}) \cdot \ln(2) 
&= -1 - \int_0^\infty w_P(h)\ln(h \cdot w_P(h)) \, dh \\
&= -1 - \int_{-\infty}^\infty e^t \cdot w_P(e^t) \cdot (e^t \cdot w_P(e^t)) \, dt \\
&= \DiffEnt{\ln H} \cdot \ln(2) - 1,
\end{align*}
where the second equality follows via the substitution $t \gets \ln h$ and the last equality follows since for the density $\ln H$ we have
\begin{align*}
\frac{d}{dt}\Prob[\ln H \leq t] = \frac{d}{dt}\Prob[H \leq e^t] = w_P(e^t) \cdot e^t
\end{align*}
\par
\textbf{Showing \cref{eq:asymptotic_result_proof_sequence_CDF_identity}:} First, by Equation 9 of \cite{flamich2024some}, we have
\begin{align*}
w_{P_{X^n}}(h) = -\int_h^\infty \frac{1}{h} \, d w_{P_{X^n \mid Y^n}}(h)
\end{align*}
Hence, we have
\begin{align}
\Prob[s(\ln H_n - b) \leq t \mid Y^n] 
&= \int_{-\infty}^t s \cdot e^{s \cdot u + b} w_{P_{X^n}}(e^{s \cdot u + b}) \, du \nonumber\\
&= -\int_{-\infty}^t s \cdot e^{s \cdot u + b} \cdot \int_{e^{s \cdot u + b}}^{\infty} \frac{1}{v} \, d w_{P_{X^n \mid Y^n}}(v) \nonumber \\
&=-\int_{-\infty}^t s \cdot e^{s \cdot u + b} \cdot \int_{u}^{\infty} e^{-(s \cdot \eta + b)} \, d \Lambda_n(\eta) \tag{substitute $v \gets \exp(s \cdot \eta + b)$}\\
&=-\int_{-\infty}^t s \cdot \int_{u}^{\infty} e^{-s \cdot (\eta - u)} \, d \Lambda_n(\eta) \nonumber\\
&=-\int_{-\infty}^t \int_{u}^{\infty} \frac{\partial}{\partial u} e^{-s \cdot (\eta - u)} \, d \Lambda_n(\eta) \nonumber\\
&=-\int_{-\infty}^t \frac{d}{du}\int_{u}^{\infty} e^{-s \cdot (\eta - u)} \, d \Lambda_n(\eta) - \int_{-\infty}^t \, d\Lambda_n(\eta) \tag{Leibniz integral rule}\\
&=\int_{-\infty}^t e^{-s(\eta - t)} \, d\Lambda_n(\eta) + 1 - \Lambda_n(t), \tag{fundamental theorem of calculus}
\end{align}
as required.
\par
\textbf{Showing \cref{lemma:lindeberg_cond}:}
The proof follows from a ``standard'' argument:
\begin{enumerate}
\item For each $n \geq 1$, we replace the infinite sum in the Lindeberg condition with a single quantity proportional to the expectation $\Exp[T_n]$ of an appropriate random variable $T_n$.
\item We show that $T_n$ converges to $0$ in probability.
\item Then, by the dominated convergence theorem, $\Exp[T_n] \to 0$ as $n \to \infty$.
\end{enumerate}
\begin{proof}
Fix $\epsilon > 0$.
Then, note that
\begin{align*}
\lim_{n \to \infty}\Exp_{Y^n}&\left[\frac{1}{s_n^2}\sum_{k = 1}^n \Exp_{X_k \mid Y_k}\left[(\ln r_{Y_k}(X_k) - \kappa_{Y_k})^2 \Ind[\abs{\ln r_{Y_k}(X_k) - \kappa_{Y_k}} \geq s_n \epsilon]\right]\right] \\
&= \lim_{n \to \infty}\sum_{k = 1}^n \Exp_{X_k, Y^n}\left[\frac{1}{s_n^2}(\ln r_{Y_k}(X_k) - \kappa_{Y_k})^2 \Ind[\abs{\ln r_{Y_k}(X_k) - \kappa_{Y_k}} \geq s_n \epsilon]\right] \\
&=\lim_{n \to \infty}n \cdot \Exp_{X_1, Y^n}\left[\frac{1}{s_n^2}(\ln r_{Y_k}(X_k) - \kappa_{Y_k})^2 \Ind[\abs{\ln r_{Y_k}(X_k) - \kappa_{Y_k}} \geq s_n \epsilon]\right],
\end{align*}
where the last line follows since the expectations over $(X_k, Y^n)$ and $(X_1, Y^n)$ are exchangeable.
Now, let
\begin{align*}
T_n = \frac{n}{s_n^2}(\ln r_{Y_k}(X_k) - \kappa_{Y_k})^2 \Ind[\abs{\ln r_{Y_k}(X_k) - \kappa_{Y_k}} \geq s_n \epsilon].
\end{align*}
By the weak law of large numbers, $s_n^2 / n \to \Exp_Y[\sigma^2_Y]$ in probability.
Hence, we also have $T_n \to 0$ in probability.
Finally, by the assumption of \cref{thm:asymptotic_result_thm_1}, $\Var_{X, Y}[\ln r_Y(X)] < \infty$, hence applying the dominated convergence theorem to $T_n$ with $(\ln r_Y(X))^2$ as the dominating variable establishes the Lindeberg condition. 
\end{proof}
\section{Details concerning the large deviations proof}
\subsection{Construction of the typical set $\mathcal{U}^{n}_{\text{typical}}$} \label{section:Typicality}
First, we consider the interval that satisfies the conditions of Lemma \ref{lemma:MomentsExistEnsemble}. Let us denote this by $[\underline{\lambda}, \overline{\lambda}]$. Then, for any $\epsilon$ that lies in this interval, let  
\begin{equation}
    K_\epsilon = \{ \underline{{\lambda}}, \underline{{\lambda}} + \epsilon, \underline{{\lambda}} + 2\epsilon, \ldots, \underline{{\lambda}} + k\epsilon, \overline{\lambda}\},
\end{equation}
where $k$ is the largest integer s.t. $\underline{{\lambda}} + k\epsilon < \overline{\lambda}$.\\

Then, for $j \in \{1,2, 3\}$ and $\lambda \in K_\epsilon$, consider the events
\begin{align}
    \Delta^{(n)}_{\epsilon} &= \left\{Y^n, Z: \left|\Exp_{X^n}\left[ \frac{1}{n}\sum\limits_{i=1}^n\gam(X_i|Y_i)|Y^n, Z \right] - \MIe{X}{Y}\right| > \epsilon \right\}, \label{eq:CodeClose}\\
    \Delta^{(n)}_{j, \lambda, \epsilon} &= \left\{ Y^n,Z: \left| \frac{1}{n}\sum\limits_{i=1}^n\Exp_{Q^\lambda_{X|Y_i}}\left[ \left| \gam(X|Y_i) \right|^j \right] - \Exp_Y\Exp_{Q^\lambda_{X|Y}}\left[ \left| \gam(X|Y) \right|^j | Y \right] \right| > \epsilon \right\} \label{eq:FiniteNet}\\
    \underline{\Delta}^{(n)}_{j, \lambda, \epsilon} &= \left\{ Y^n,Z: \left| \frac{1}{n}\sum\limits_{i=1}^n\Exp_{Q^\lambda_{X|Y_i}}\left[ \left( \gam(X|Y_i) \right)^j \right] - \Exp_Y\Exp_{Q^\lambda_{X|Y}}\left[ \left( \gam(X|Y) \right)^j | Y \right] \right| > \epsilon \right\} \label{eq:FiniteNetRaw}
\end{align}
Let $\mathcal{U}^{n, \epsilon}_{\text{atypical}}$ denote the union of the above events for all values of $j$ and $\lambda$. Then, we can show that the probability of $\mathcal{U}^{n, \epsilon}_{\text{atypical}}$ decays as $O\left(\frac{1}{n}\right)$:
\begin{lemma} \label{lemma:TypicalityProbBound}
    For all $n>0$, we have
    \begin{equation}
        P_{Y^n,Z}(\mathcal{U}^{n, \epsilon}_{\text{atypical}}) \le \frac{C}{n},
    \end{equation}
    for a constant $C$ that is a function of $P_{XY}$ and $\epsilon$.
\end{lemma}
\begin{proof}
    Please see Appendix \ref{section:TypicalityProbBound}.
\end{proof}
We can then refer to the $\epsilon$-typical set $\mathcal{U}^{n, \epsilon}_{\text{typical}} = \left(\mathcal{U}^{n, \epsilon}_{\text{atypical}}\right)^c$.
The notion of typically we imposed we imposed allow us to prove some strong regularity conditions for the tilted moments that we will need to prove large deviations results.
\begin{lemma} \label{lemma:regularity}
For all sufficiently small values of $\epsilon>0$, there exist positive constants $\underline{m}_2, \overline{m}_2, \overline{m}_3$ s.t. for all $n$, all $\lambda \in [\underline{\lambda}, \overline{\lambda}]$ and all $\epsilon$-typical $(y^n,z)$, the following hold:
\begin{enumerate}
    \item $\underline{m}_2 \le \frac{s^2_n(\lambda, y^n)}{n} \le \overline{m}_2$.
    \item $\frac{\muthird(\lambda, y^n)}{n} \le \overline{m}_3$.
\end{enumerate}
\end{lemma}
\begin{proof}
    Please see Appendix~\ref{section:RegProof}
\end{proof}

\begin{definition} \label{def:EpsilonChoice}
    We will choose a value of $\epsilon$ that ensures that the results of Lemma \ref{lemma:regularity} hold along with
    \begin{align}
        \MIe{X}{Y} + 3\epsilon &\le \Exp_Y\Exp_{Q^{\overline{\lambda}}_{X|Y}}\left[ \gam(X|Y) | Y \right],\\
        \MIe{X}{Y} - 3\epsilon &\ge \Exp_Y\Exp_{Q^{\underline{\lambda}}_{X|Y}}\left[ \gam(X|Y) | Y \right].
    \end{align}
    All typical $(y^n, z)$ henceforth will be assumed to be typical w.r.t. this choice of $\epsilon$, and we therefore drop the superscript $\epsilon$ going forward.
\end{definition}

\begin{definition} \label{def:YucelLB}
    For any $\lambda > 0, s>0, \mu \ge 0, n>0$, we define $\underline{M}_n(\lambda, s, \mu)$ to be
    \begin{align}
        \underline{M}_n(\lambda, s, \mu) &= \frac{(1+2t)\gamma(\lambda, s, \mu)\sqrt{n}}{2\lambda\sqrt{2\pi}s\exp(2t)}, \text{ where}\\
        \gamma(\lambda, s, \mu) &= 1 - \frac{1 + (1 + 2t)^2}{\lambda(1 + 2t)\sqrt{e}s}, \text{ and}\\
        t &= \lambda 2\sqrt{2\pi}\frac{\mu}{s^2}.\\
    \end{align}
    Let $n_0$ denote the threshold for $n$ s.t. $\gamma(\lambda, s, \mu) > \frac{1}{2}$ for all $\frac{s^2}{n} \in (\underline{m}_2, \overline{m}_2)$, $\frac{\mu}{n} \in (0, \overline{m}_3)$ and $\lambda \in [\underline{\lambda}, \overline{\lambda}]$. Then, as a corollary of Lemma \ref{lemma:regularity}, we know that $\underline{M}_n(\lambda, s(\lambda, y^n), \muthird(\lambda, y^n))$ is lower bounded by a positive constant for all $n>n_0$, $\lambda \in [\underline{\lambda}, \overline{\lambda}]$ and typical $(y^n,z)$. Let us denote this constant by $\underline{M}$.
\end{definition}

\subsection{Auxiliary Results}

\begin{lemma} \label{lemma:MomentsExistEnsemble}
    There exists an interval $[\underline{\lambda}, \beta]$ where $\underline{\lambda} < 1 < \beta \le 1 + \delta$ and the following hold:
    \begin{enumerate}
        \item  
        \begin{equation}
            \Exp_Y\Exp_{Q^\lambda_{X|Y}}\left[ \left| \gam(X|Y) \right|^k \Big| Y \right] < \infty, \text{ for all } k \in \{1,2,3,4,5,6\} \text{ and }\lambda \in [\underline{\lambda}, \beta].\label{eq:IntervalExist1}
        \end{equation}
        \item 
        \begin{equation}
            \inf\limits_{[\underline{\lambda}, \beta]} \Exp_Y\Var_{Q^\lambda_{X|Y}}\left[ \gam(X|Y) \Big| Y \right] > 0. \label{eq:IntervalExist2}
        \end{equation}
    \end{enumerate}
\end{lemma}
\begin{proof}
    Both the conclusions follow from a continuity argument. 
    
    Without loss of generality, assume $\delta < 1$. For any arbitrary random variable $\Theta$ s.t. $\Exp_{XY}\left[ \Theta^2 \right] < \infty$ and $\lambda > 0$, consider
    \begin{align}
        &\; \lim\limits_{\lambda \to 1} J(\lambda)\\
        &= \lim\limits_{\lambda \to 1} \left| \Exp_{P_YQ^1_{X|Y}}\left[ \Theta \right] - \Exp_{P_YQ^\lambda_{X|Y}}\left[ \Theta \right] \right|^2\\
        &= \lim\limits_{\lambda \to 1} \left| \Exp_{XY}\left[ \Theta \right] - \Exp_{XY}\left[ \frac{dQ^\lambda_{X|Y}}{dP_{X|Y}} \Theta \right] \right|^2\\
        &= \lim\limits_{\lambda \to 1} \left| \Exp_{XY}\left[ \left( 1 - \exp\left( (\lambda - 1) \gam - \empLam(\lambda, Y) \right) \right) \Theta \right] \right|^2\\
        &\le \lim\limits_{\lambda \to 1} \;\Exp_{XY}\left[ \left( 1 - \exp\left( (\lambda - 1) \gam - \empLam(\lambda, Y) \right) \right)^2 \right] \Exp_{XY}\left[ \Theta^2\right] \label{eq:CSBound}\\
        &= \lim\limits_{\lambda \to 1} \;\Exp_{XY}\left[ 1 + \exp\left( 2(\lambda - 1) \gam - 2\empLam(\lambda, Y) \right) - 2\exp\left( (\lambda - 1) \gam - \empLam(\lambda, Y) \right) \right] \times \nonumber\\
        &\phantom{\lim\limits_{\lambda \to 1} \Exp_ff} \Exp_{XY}\left[ \Theta^2\right]. \label{eq:Split_terms} 
    \end{align}
    We can bound the two $\lambda$ dependant expectations in (\ref{eq:Split_terms}) separately. First, consider
    \begin{align}
        &\lim\limits_{\lambda \to 1} \;\Exp_{XY}\left[ 2\exp\left( (\lambda - 1) \gam - \empLam(\lambda, Y) \right) \right]\\
        &\ge 2\exp\left( \lim\limits_{\lambda \to 1} \Exp_{XY}\left[ (\lambda - 1) \cdot  \gam - \empLam(\lambda, Y)  \right]\right) \label{eq:JensenGe}\\
        &= 2\exp\left( - \lim\limits_{\lambda \to 1} E_Y\left[\empLam(\lambda, Y) \right] \right). \label{eq:SecondExpectLB'}
    \end{align}
    Here, (\ref{eq:JensenGe}) follows from applying Jensen's inequality. If we consider all $\lambda$ within $[1 - \delta, 1 + \delta]$, from the convexity of $\empLam(\lambda, y)$ (See Lemma 2.2.5 \cite{dembo2009large}) for all $y$, we can obtain
    \begin{align}
        \empLam(\lambda, y) \le \empLam(1 + \delta, y) + \empLam(1 - \delta, y).
    \end{align}
    Therefore, we can apply the dominated convergence theorem to (\ref{eq:SecondExpectLB'}) to take the limit inside the expectation to obtain
    \begin{align}
        \lim\limits_{\lambda \to 1} \;\Exp_{XY}\left[ 2\exp\left( (\lambda - 1) \gam - \empLam(\lambda, Y) \right) \right] = 2. \label{eq:SecondExpectLB}
    \end{align}
    Next, consider the other expectation from (\ref{eq:Split_terms}):
    \begin{align}
        & \lim\limits_{\lambda \to 1} \Exp_{XY} \left[\exp\left( 2(\lambda - 1) \gam - 2\empLam(\lambda, Y) \right) \right]\\
        &\le \sqrt{\lim\limits_{\lambda \to 1}\Exp_{XY}\left[ \exp\left( 4(\lambda - 1) \gam \right)\right]\Exp_{Y}\left[ \exp\left( -4\empLam(\lambda, Y)  \right)\right]} \label{eq:CSagain}\\
        &\le \sqrt{\lim\limits_{\lambda \to 1}\Exp_{XY}\left[ \exp\left( 4(\lambda - 1) \gam \right)\right]\Exp_{Y}\left[ \exp\left(  -4(\lambda - 1)\empLam'(1, Y)  \right)\right]}
        \label{eq:DC}
    \end{align}
    Here, (\ref{eq:CSagain}) follows from the Cauchy-Schwarz inequality. To obtain (\ref{eq:DC}), we have used the convexity of $\empLam$:
    \begin{align}
        \empLam(\lambda,y) &\ge \empLam(1,y) + (\lambda - 1)\empLam'(1, y)\\
        &= (\lambda - 1)\empLam'(1, y), \text{ for all } y.
    \end{align}
    We can then show that we can take the limit inside the expectation for both the terms in (\ref{eq:DC}) by the dominated convergence theorem. If we consider $\lambda \in \left( 1 - \frac{\delta}{4}, 1 + \frac{\delta}{4}\right)$, we can obtain the following uniform bounds
    \begin{align}
        \exp\left( 4(\lambda - 1) \gam(x|y) \right) &\le \exp\left( \delta \gam(x|y) \right) + \exp\left( -\delta \gam(x|y) \right),\text{ for all }x,y, \label{eq:firstDC}\\
        \exp\left( -4(\lambda - 1)\empLam'(1, y)  \right) &\le \exp\left( \delta \empLam'(1, y) \right) + \exp\left( -\delta \empLam'(1, y) \right),\text{ for all } y. \label{eq:secondDC}
    \end{align}
    The right hand side of (\ref{eq:firstDC}) is $P_{XY}$-integrable by definition. For (\ref{eq:secondDC}), using Jensen's inequality, we obtain
    \begin{align}
        &\Exp_Y\left[ \exp\left( \delta\empLam'(1, Y) \right)\right]\\
        &= \Exp_Y\left[ \exp\left( \delta\Exp_{X|Y}\left[ \gam \Big| Y\right] \right)\right]\\
        &\le \Exp_{XY}\left[ \exp\left( \delta\gam \right) \right] < \infty,
    \end{align}
    and similarly for $-\delta$. Thus, we can use dominated convergence to evaluate both limits, resulting in
    \begin{align}
        \lim\limits_{\lambda \to 1} \Exp_{XY} \left[\exp\left( 2(\lambda - 1) \gam - 2\empLam(\lambda, Y) \right) \right] = 1. \label{eq:I1}
    \end{align}
   Combining the results from (\ref{eq:SecondExpectLB}) and (\ref{eq:I1}) with (\ref{eq:Split_terms}), we obtain $\lim\limits_{\lambda \to 1}J(\lambda) = 0$.
   
    We have shown that $\Exp_{P_YQ^\lambda_{X|Y}}\left[ \Theta \right]$ is continuous at $\lambda = 1$. Taking $\Theta$ to be $\left|\gam\right|^k$ allows us to find $k$ intervals that satisfy (\ref{eq:IntervalExist1}) and taking $\Theta$ to be $\left(\gam - \Exp_{X|Y}\left[ \gam \Big| Y \right]\right)^2$ \footnote{For non-singular channels, $\Exp_Y\Var_{X|Y}\left[ \gam(X|Y) \Big| Y \right] > 0$.} allows to to similarly find an interval that satisfies (\ref{eq:IntervalExist2}). Choosing the smallest of these intervals suffices to complete the proof. 
\end{proof}

\begin{lemma} \label{lemma:StochasticDominance}
    For any $y \in \mathcal{Y}$, and $\lambda_1, \lambda_2 \in [\underline{\lambda}, \overline{\lambda}]$ s.t. $\lambda_2 > \lambda_1$, $\gam(X|y)$ under $Q^{\lambda_2}_{X|y}$ stochastically dominates $\gam(X|y)$ under $Q^{\lambda_1}_{X|y}$.
\end{lemma}
\begin{proof}
    For any $a>0$, consider 
    \begin{align}
        &\frac{Q^{\lambda_2}_{X|y}\left( \gam(X|y) \ge a\right)}{1 - Q^{\lambda_2}_{X|y}\left( \gam(X|y) \ge a\right)} \\
        &= \frac{\Exp_{Q^{\lambda_2}_{X|y}}\left[\mathbf{1}\left(\gam(X|y) \ge a\right)\right]}{ \Exp_{Q^{\lambda_2}_{X|y}}\left[\mathbf{1}\left(\gam(X|y) < a\right)\right]} \\
        &= \frac{\Exp_{P_X}\left[\mathbf{1}\left(\frac{dP_{X|Y}}{dP_X}(X|y) > 0\right)\exp\left(\lambda_2\gam(X|y) - \Lambda_X(\lambda_2, y)\right) \mathbf{1}\left(\gam(X|y) \ge a\right)\right]}{\Exp_{P_X}\left[\mathbf{1}\left(\frac{dP_{X|Y}}{dP_X}(x|y) > 0\right)\exp\left(\lambda_2 \gam(X|y) - \Lambda_X(\lambda_2, y)\right)\mathbf{1}\left(\gam(X|y) < a\right)\right]}\\
        &= \frac{\Exp_{P_X}\left[\mathbf{1}\left(\frac{dP_{X|Y}}{dP_X}(X|y) > 0\right)\exp\left((\lambda_2 - \lambda_1)\gam(X|y) \right)\exp\left( \lambda_1 \gam(X|y) \right) \mathbf{1}\left( \gam(X|y) \ge a\right)\right]}{\Exp_{P_X}\left[\mathbf{1}\left(\frac{dP_{X|Y}}{dP_X}(X|y) > 0\right)\exp\left((\lambda_2 - \lambda_1)\gam(X|y) \right)\exp\left( \lambda_1 \gam(X|y) \right) \mathbf{1}\left(\gam(X|y) < a\right)\right]}\\
        &\ge \frac{\Exp_{P_X}\left[\mathbf{1}\left(\frac{dP_{X|Y}}{dP_X}(X|y) > 0\right)\exp\left((\lambda_2 - \lambda_1)a \right)\exp\left( \lambda_1 \gam(X|y) \right) \mathbf{1}\left(\gam(X|y) \ge a\right)\right]}{\Exp_{P_X}\left[\mathbf{1}\left(\frac{dP_{X|Y}}{dP_X}(X|y) > 0\right)\exp\left((\lambda_2 - \lambda_1)a \right)\exp\left( \lambda_1 \gam(X|y) \right) \mathbf{1}\left(\gam(X|y) < a\right)\right]}\\
        &= \frac{\Exp_{P_X}\left[\mathbf{1}\left(\frac{dP_{X|Y}}{dP_X}(X|y) > 0\right)\exp\left( \lambda_1 \gam(X|y) \right) \mathbf{1}\left( \gam(X|y) \ge a\right) \right]}{\Exp_{P_X}\left[\mathbf{1}\left(\frac{dP_{X|Y}}{dP_X}(X|y) > 0\right)\exp\left( \lambda_1 \gam(X|y) \right) \mathbf{1}\left(\gam(X|y) < a\right)\right]}\\
        &= \frac{Q^{\lambda_1}_{X|y}\left( \gam(X|y) \ge a\right)}{1 - Q^{\lambda_1}_{X|y}\left( \gam(X|y) \ge a\right)}.
    \end{align}
    This implies that $Q^{\lambda_2}_{X|y}\left( \gam(X|y) \ge a\right) \ge Q^{\lambda_1}_{X|y}\left( \gam(X|y) \ge a\right)$ because the function $\frac{z}{1-z}$ is strictly increasing on its domain.
\end{proof}
\begin{corollary}
    For any $y \in \mathcal{Y}$, $\lambda_1, \lambda_2 \in [\underline{\lambda}, \overline{\lambda}]$ s.t. $\lambda_2 \ge \lambda_1$ and non-decreasing function $h$,
    \begin{equation}
        \Exp_{Q^{\lambda_2}_{X|y}}\left[h\left(\gam(X|y)\right)\right] \ge \Exp_{Q^{\lambda_1}_{X|y}}\left[h\left(\gam(X|y)\right)\right].
    \end{equation}
\end{corollary}
\begin{proof}
    This result follows from Lemma \ref{lemma:StochasticDominance} and the properties of stochastic dominance (see Proposition 6.D.1 in \cite{andreu1995microeconomic}).
\end{proof}

\begin{lemma} \label{lemma:MomentBounds}
    For any $k>0$ and $y \in \mathcal{Y}$,
    \begin{equation}
        \sup\limits_{\lambda \in [\underline{\lambda}, \overline{\lambda}]}\Exp_{Q^{\lambda}_{X|y}}\left[ \left|\gam(X|y)\right|^k\right] \le \Exp_{Q^{\overline{\lambda}}_{X|y}}\left[ \left|\gam(X|y) \right| ^k \right] + \Exp_{Q^{\underline{\lambda}}_{X|y}}\left[\left|\gam(X|y) \right| ^k\right].
    \end{equation}
\end{lemma}
\begin{proof}
    We have
    \begin{align}
        &\sup\limits_{\lambda \in [\underline{\lambda}, \overline{\lambda}]}\Exp_{Q^{\lambda}_{X|y}}\left[ \left|\gam(X|y)\right|^k\right]\\
        &=\sup\limits_{\lambda \in [\underline{\lambda}, \overline{\lambda}]} \Exp_{Q^{\lambda}_{X|y}}\left[ \max\left( 0, \gam(X|y) \right)^k + \max\left( 0, -\gam(X|y) \right)^k\right]\\
        &\le \sup\limits_{\lambda \in [\underline{\lambda}, \overline{\lambda}]} \Exp_{Q^{\lambda}_{X|y}}\left[ \max\left( 0, \gam(X|y) \right)^k \right] + \sup\limits_{\lambda \in [\underline{\lambda}, \overline{\lambda}]} \Exp_{Q^{\lambda}_{X|y}}\left[\max\left( 0, -\gam(X|y) \right)^k\right].
    \end{align}
    We can then apply the result from Lemma \ref{lemma:StochasticDominance} to complete the proof
    \begin{align}
        &\sup\limits_{\lambda \in [\underline{\lambda}, \overline{\lambda}]}\Exp_{Q^{\lambda}_{X|y}}\left[ \left|\gam(X|y)\right|^k\right]\\
        &\le \Exp_{Q^{\overline{\lambda}}_{X|y}}\left[ \max\left( 0, \gam(X|y) \right)^k \right] + \Exp_{Q^{\underline{\lambda}}_{X|y}}\left[\max\left( 0, -\gam(X|y) \right)^k\right]\\
        &\le \Exp_{Q^{\overline{\lambda}}_{X|y}}\left[ \left|\gam(X|y) \right| ^k \right] + \Exp_{Q^{\underline{\lambda}}_{X|y}}\left[\left|\gam(X|y) \right| ^k\right].
    \end{align}
\end{proof}

\subsection{Proofs of Lemmas}

\subsubsection{Lemma \ref{lemma:TypicalityProbBound}} \label{section:TypicalityProbBound}
\begin{proof}
    By Lemma \ref{lemma:MomentBounds}, we know that $\Exp_Y\Exp_{Q^\lambda_{X|Y}}\left[ \left| \gam(X|Y) \right|^k | Y \right] < \infty$ for every $k \in \{1,2,3,4,5,6\}$. Hence, each of the atypical sets defined in (\ref{eq:CodeClose}) - (\ref{eq:FiniteNetRaw}) are concentration events for random variables with finite mean and variance and therefore, their probabilities can be bounded using Chebyshev's inequality. For $\Delta^{(n)}_\epsilon$, we can obtain
\begin{align}
    P_{Y^n,Z}(\Delta^{(n)}_\epsilon) &\le \frac{\Var\left[\Exp_X\left[ \sum\limits_{i=1}^n\gam(X_i|Y_i)|Y^n, Z \right]\right]}{n^2\epsilon^2}\\
    &\le \frac{\Var\left[ \gam \right]}{n\epsilon^2}, \label{eq:Chebyshev2}
\end{align}
and for each $\Delta^{(n)}_{j, \lambda, \epsilon}$,
\begin{align}
    P_{Y^n,Z}\left( \Delta^{(n)}_{j, \lambda, \epsilon} \right) &\le \frac{\Var\left[\Exp_{Q^\lambda_{X^n|Y^n}}\left[ \sum\limits_{i=1}^n \left|\gam(X_i|Y_i)\right|^j|Y^n \right]\right]}{n^2\epsilon^2} \label{eq:chebBegin}\\
    &\le \frac{ \Var _{P_YQ^\lambda_{X|Y}}\left[ \left|\gam(X_i|Y_i)\right|^j \right] }{n\epsilon^2}\\
    &\le \frac{ \Exp_{P_YQ^\lambda_{X|Y}}\left[ \left|\gam(X_i|Y_i)\right|^{2j} \right] }{n\epsilon^2}.\label{eq:ChebyShev4}
\end{align}
The same bound holds for each $P_{Y^n,Z}\left( \underline{\Delta}^{(n)}_{j, \lambda, \epsilon} \right)$ as well. Applying the union bound then completes the proof.
\end{proof}

\subsubsection{Lemma \ref{lemma:regularity}} \label{section:RegProof}
\begin{proof}
   We first show the existence of the upper bounds --- these follow directly from Lemma \ref{lemma:MomentBounds} and the typicality conditions in (\ref{eq:FiniteNet}). For $k \in \{1,2,3\}$, consider
   \begin{align}
       &\sup\limits_{\lambda \in [\underline{\lambda}, \overline{\lambda}]} \frac{1}{n}\sum\limits_{i=1}^n\Exp_{Q^\lambda_{X|y_i}}\left[ \left| \gam(X|y_i) \right|^k \right] \label{eq:BeginRawMomentBound}\\
       &\le  \frac{1}{n}\sum\limits_{i=1}^n\Exp_{Q^{\underline{\lambda}}_{X|y_i}}\left[ \left| \gam(X|y_i) \right|^k \right] +  \frac{1}{n}\sum\limits_{i=1}^n\Exp_{Q^{\overline{\lambda}}_{X|y_i}}\left[ \left| \gam(X|y_i) \right|^k \right]\\
       &\le \Exp_Y\Exp_{Q^{\underline{\lambda}}_{X|Y}}\left[ \left| \gam(X|Y) \right|^j | Y \right] + \Exp_Y\Exp_{Q^{\overline{\lambda}}_{X|Y}}\left[ \left| \gam(X|Y) \right|^j | Y \right] + 2\epsilon. \label{eq:RawMomentsRegularityUsed}
   \end{align}
   Now, for $\frac{s^2_n}{n}$, we have
   \begin{align}
       &\frac{s^2_n}{n}\\
       &=\frac{1}{n}\sum\limits_{i=1}^n \Var_{Q^{\lambda}_{X|y_i}}\left[  \gam(X|y_i) \right]\\
       &= \frac{1}{n}\sum\limits_{i=1}^n \left( \Exp_{Q^{\lambda}_{X|y_i}}\left[  \left(\gam(X|y_i)\right)^2 \right] - \Exp_{Q^{\lambda}_{X|y_i}}\left[  \gam(X|y_i) \right]^2\right)\\
       &\le \frac{1}{n}\sum\limits_{i=1}^n \left( \Exp_{Q^{\lambda}_{X|y_i}}\left[  \left|\gam(X|y_i)\right|^2 \right] + \Exp_{Q^{\lambda}_{X|y_i}}\left[  \left| \gam(X|y_i) \right| \right]^2\right).
   \end{align}
   Applying (\ref{eq:BeginRawMomentBound}) - (\ref{eq:RawMomentsRegularityUsed}) with $k=1$ and $k=2$ then gives us $\overline{m}_2$. Similarly, for $\frac{1}{n}\sum\limits_{i=1}^n \Lambda_X'''(\lambda,y_i)$, we can obtain
   \begin{align}
       &\frac{1}{n}\sum\limits_{i=1}^n \Lambda_X'''(\lambda,y_i)\\
       &=\frac{1}{n}\sum\limits_{i=1}^n \Exp_{Q^{\lambda}_{X|y_i}}\left[  \left(\gam(X|y_i) - \Exp_{Q^{\lambda}_{X|y_i}}\left[ \gam \right]\right)^3 \right]\\
       &= \frac{1}{n}\sum\limits_{i=1}^n \Bigg( \Exp_{Q^{\lambda}_{X|y_i}}\left[\left(\gam(X|y_i)\right)^3\right] - 3\Exp_{Q^{\lambda}_{X|y_i}}\left[\gam(X|y_i)\right]\Exp_{Q^{\lambda}_{X|y_i}}\left[\left(\gam(X|y_i)\right)^2\right] \nonumber\\
       &\phantom{\sum\limits_{i=1}^n}+ 2\Exp_{Q^{\lambda}_{X|y_i}}\left[\gam(X|y_i)\right]^3\Bigg)\\
       &\le \frac{1}{n} \sum\limits_{i=1}^n \Bigg( \Exp_{Q^{\lambda}_{X|y_i}}\left[\left|\gam(X|y_i)\right|^3\right] + 3\Exp_{Q^{\lambda}_{X|y_i}}\left[\left|\gam(X|y_i)\right|\right]\Exp_{Q^{\lambda}_{X|y_i}}\left[\left|\gam(X|y_i)\right|^2\right] \nonumber\\
       &\phantom{\sum\limits_{i=1}^n}+ 2\Exp_{Q^{\lambda}_{X|y_i}}\left[\left|\gam(X|y_i)\right|\right]^3\Bigg).
   \end{align}
   We can then similarly find $\overline{m}_3$ by applying (\ref{eq:BeginRawMomentBound}) - (\ref{eq:RawMomentsRegularityUsed}). Note that $\overline{m}_3$ is an upper bound for the tilted central third moment as well as $\frac{\muthird}{n}$.

   For the lower bound on $\frac{s^2_n(\lambda, y^n)}{n}$, we will use a continuity-based argument that relies on the fact that $\Exp_Y\Var_{Q^{\eta}_{X|Y}}\left[ \left( \gam(X|Y) \right) | Y \right]$ can be uniformly bounded for $\eta \in [\underline{\lambda}, \overline{\lambda}]$. Let $\hat{\lambda}$ denote the closest point in $K_\epsilon$ to $\lambda$. Then, by Taylor's theorem,
   \begin{align}
       \frac{s^2_n(\lambda, y^n)}{n} &= \frac{1}{n}\sum\limits_{i=1}^n \empLam''(\lambda, y^n)\\
       &= \frac{1}{n}\sum\limits_{i=1}^n \empLam''(\hat{\lambda}, y^n) + \frac{\lambda - \hat{\lambda}}{n}\sum\limits_{i=1}^n \empLam'''(\tilde{\lambda}, y^n),
   \end{align}
   for some $\tilde{\lambda}$ in the interval between $\lambda$ and $\hat{\lambda}$. Using the typicality condition in (\ref{eq:FiniteNet}), and the upper bound on the tilted third central moment $\overline{m}_3$, we can then obtain the following lower bound
   \begin{align}
       &\frac{s^2_n(\lambda, y^n)}{n} \\
       &\ge \Exp_Y\Var_{Q^{\hat{\lambda}}_{X|Y}}\left[ \left( \gam(X|Y) \right) | Y \right] - 2\epsilon - \overline{m}_3\epsilon\\
       &\ge \inf \limits_{\lambda_0 \in [\underline{\lambda}, \overline{\lambda}]} \Exp_Y\Var_{Q^{\lambda_0}_{X|Y}}\left[ \left( \gam(X|Y) \right) | Y \right] - 2\epsilon - \overline{m}_3\epsilon. \label{eq:lbFinal}
   \end{align}
   Choosing $\epsilon$ to be sufficiently small ensures that the expression in (\ref{eq:lbFinal}) is bounded away from $0$, thus obtaining the lower bound $\underline{m}_2$.
\end{proof}

\subsection{Proof of Lemma \ref{lemma:Gibbs}} \label{section:GibbsProof}
\begin{proof}
    Choose $n>\max\left(n_0, \frac{1}{\underline{M} \cdot \underline{\lambda}^2 \cdot \sqrt{\underline{m}_2} \cdot \epsilon}\right)$ (See Definition \ref{def:YucelLB} for $n_0$). Then, from the typicality condition from (\ref{eq:FiniteNetRaw}) for $j=1$ and the choice of $\epsilon$ (See Definition \ref{def:EpsilonChoice}), we obtain
    \begin{align}
        \empLam'(\underline{\lambda}, y^n) < \Exp_Y\Exp_{Q^{\underline{\lambda}}_{X|Y}}\left[ \gam(X|Y) | Y \right] + \epsilon < \imath_{y^n,z} < \Exp_Y\Exp_{Q^{\overline{\lambda}}_{X|Y}}\left[ \gam(X|Y) | Y \right] - \epsilon < \empLam'(\overline{\lambda}, y^n)
    \end{align}
    Hence, we can find $\lambda \in (\underline{\lambda}, \overline{\lambda})$ s.t. $\empLam'(\lambda, y^n) = \imath_{y^n,z}$. Then, we have
    \begin{align}
        \imath_{y^n}(B(y^n,z)) &= \Exp_{p_{X^n}}\left[\gamyn \Bigg| \gamyn \ge \imath_{y^n,z}\right]\\
        &= \frac{\Exp_{p_{X^n}}\left[\gamyn \cdot \mathbf{1}\left( \gamyn \ge \imath_{y^n,z} \right)\right]}{p_{X^n}\left( \gamyn \ge \imath_{y^n,z} \right)}.
    \end{align}
    Transforming to the $\lambda$-tilted measure $Q^\lambda_{X^n|y^n}$ and substituting 
    \begin{align}
        W_n &= \frac{\sum\limits_{i=1}^n \gam(X_i|y_i) - n\imath_{y^n,z}}{s_n},\\
        \psi_n &= \lambda s_n,
    \end{align}
    we obtain
    \begin{align}
        &\frac{\Exp_{p_{X^n}}\left[\gamyn \cdot \mathbf{1}\left( \gamyn \ge \imath_{y^n,z} \right)\right]}{p_{X^n}\left( \gamyn \ge \imath_{y^n,z} \right)}\\
        &= \frac{\Exp_{Q^\lambda_{X^n|y^n}}\left[ \left( \imath_{y^n,z} + \frac{s_nW_n} {n}\right)\exp\left( -\psi_n W_n\right) \mathbf{1}\left( W_n \ge 0 \right)\right]}{ \Exp_{Q^\lambda_{X^n|y^n}}\left[ \exp\left( -\psi_n W_n\right) \mathbf{1}\left( W_n \ge 0\right) \right]} \\
        &= \imath_{y^n,z} + \frac{s_n}{n}\frac{\Exp_{Q^\lambda_{X^n|y^n}}\left[ W_n \exp\left( -\psi_n W_n\right) \mathbf{1}\left( W_n \ge 0 \right)\right]}{\Exp_{Q^\lambda_{X^n|y^n}}\left[ \exp\left( -\psi_n W_n\right)\mathbf{1}\left( W_n \ge 0 \right)\right]}. \label{eq:gibbsSubs}
    \end{align}
    If we denote the law of $W_n$ by $F_n(\cdot)$, we have
    \begin{align}
        &\Exp_{Q^\lambda_{X^n|y^n}}\left[ W_n \exp\left( -\psi_n W_n\right)\mathbf{1}\left( W_n \ge 0 \right)\right]\\ &=\int\limits_{0}^\infty x\exp\left( -\psi_n x\right) dF_n\\
        &=  \int\limits_{0}^\infty \left( \psi_n x\exp\left( -\psi_n x\right) - \exp\left( -\psi_n x\right) \right) F_n(x)dx \label{eq:IntegrateByParts11}\\
        &= \frac{1}{\psi_n} \int\limits_{0}^\infty \left(t\exp\left( -t \right) - \exp(-t) \right) F_n\left( \frac{t}{\psi_n}\right)dt.
    \end{align}
    Here, (\ref{eq:IntegrateByParts11}) follows from integration by parts. Using the Berry-Esseen theorem (see (1), \cite{van1972application}),
    \begin{align}
        &\frac{1}{\psi_n} \int\limits_{0}^\infty \left(t\exp\left( -t \right) - \exp(-t) \right) F_n\left( \frac{t}{\psi_n}\right)dt\\
        &\le \frac{1}{\psi_n} \int\limits_{0}^\infty \left(t\exp\left( -t \right) - \exp(-t) \right) \left( \Phi\left( \frac{t}{\psi_n} \right) + \frac{2\mu_n^{(3)}}{s_n^3} \right)dt\\
        &= \frac{1}{\psi_n} \int\limits_{0}^\infty \left(t\exp\left( -t \right) - \exp(-t) \right) \left( \Phi\left( \frac{t}{\psi_n} \right) - \Phi(0) \right)dt\\
        &\le\frac{1}{\psi_n} \int\limits_{0}^\infty \left(t\exp\left( -t \right) - \exp(-t) \right) \left( \frac{t}{\psi_n}\phi(0) \right)dt\\
        &\le \frac{1}{\psi_n^2}.
    \end{align}
    Substituting this back into (\ref{eq:gibbsSubs}), and using the lower bound from Lemma 1 in \cite{altuug2020exact}, we obtain the required result.
    \begin{align}
        &\imath_{y^n,z} + \frac{s_n}{n}\frac{\Exp_{Q^\lambda_{X^n|y^n}}\left[ W_n \exp\left( -\psi_n W_n\right) \mathbf{1}\left( W_n \ge 0 \right)\right]}{\Exp_{Q^\lambda_{X^n|y^n}}\left[ \exp\left( -\psi_n W_n\right)\mathbf{1}\left( W_n \ge 0 \right)\right]}\\
        &\le \imath_{y^n,z} + \frac{s_n}{n}\frac{\frac{1}{\psi_n^2}}{\underline{M}\frac{1}{\sqrt{n}}}\\
        &\le \imath_{y^n,z} + \frac{1}{\underline{M} \cdot \underline{\lambda}^2 \cdot \sqrt{\underline{m}_2} \cdot n}\\
        &= \imath_{y^n}(A(y^n,z)).
    \end{align}
\end{proof}

\end{document}